\documentclass[journal]{IEEEtran}

\IEEEoverridecommandlockouts                              %

\usepackage{graphicx}
\usepackage{graphics}
\usepackage{amssymb}
\usepackage{amsmath}
\usepackage{color}
\usepackage{xspace}
\usepackage{algpseudocode}
\usepackage{bbm}
\usepackage{comment}
\usepackage{algorithm}
\usepackage{algorithmicx}
\usepackage{subfig}
\usepackage{soul}
\usepackage{cite}
\usepackage{array}
\usepackage[bottom]{footmisc}
\usepackage{tikz}
\usetikzlibrary{shapes,arrows}
\usetikzlibrary{decorations.pathreplacing,decorations.markings,decorations.pathmorphing,decorations.shapes}
\usetikzlibrary{positioning,fit}
\usetikzlibrary{calc}
\usepackage{multirow}
\usepackage{cancel}
\usepackage{hyperref} 
\usepackage{xr}

\tikzstyle{block}=[draw opacity=0.7,line width=1.4cm]

\definecolor{CranJ}{cmyk}{0,0.69,0.54,0.04} 
\definecolor{PinkJ}{cmyk}{0,0.71,0.43,0.12} 
\definecolor{Cran}{cmyk}{0,0.73,0.41,0.29} 
\definecolor{VRed}{cmyk}{0,0.75,0.25,0.2} 
\definecolor{ORed}{cmyk}{0,0.75,0.75,0} 
\definecolor{CBlue}{cmyk}{1,0.25,0,0} 
\setstcolor{VRed}
\usepackage{bm}

\newlength\myindent


\tikzset{cloud/.pic={
\node[cloud, cloud puffs=10.8,cloud puff arc=110, aspect=2, draw, text width=3cm
    ] () at (0,0) {\tikzpictext};
}}

\title{\Large \bf An IMM-based Decentralized Cooperative Localization with \\LoS and NLoS UWB Inter-agent Ranging}

\author{Jianan Zhu \quad Solmaz S. Kia, \emph{senior member, IEEE }\\
\normalsize{\emph{University of California Irvine}}
  \thanks{The authors
    are with the Mechanical and Aerospace Eng. Dept. of 
    Univ. of California Irvine, CA 92697,~USA, {\tt\small jiananz1,solmaz@uci.edu}. This work is supported by the U.S. Dept. of Commerce, National Institute of Standards and Technology award 70NANB17H192.
    }%
}

\newcommand{\real}{{\mathbb{R}}}

\newcommand{\argmin}{\operatorname{argmin}}

\newcommand{\prpg}{\mbox{{\footnotesize\textbf{--}}}}
\newcommand{\updt}{\mbox{\textbf{+}}}

\newcommand{\vect}[1]{\boldsymbol{\mathbf{#1}}}
\newcommand{\Bvect}[1]{\bm\bar{\boldsymbol{\mathbf{#1}}}}
\newcommand{\Tvect}[1]{\bm\tilde{\boldsymbol{\mathbf{#1}}}}
\newcommand{\Hvect}[1]{\bm\hat{\boldsymbol{\mathbf{#1}}}}

 \newcommand{\boxend}{\hfill \ensuremath{\Box}}
\newcommand{\oprocendsymbol}{\hbox{$\bullet$}}
\newcommand{\oprocend}{\relax\ifmmode\else\unskip\hfill\fi\oprocendsymbol}

\allowdisplaybreaks

\newtheorem{rem}{Remark}[section]

\newtheorem{lem}{Lemma}[section]

\newenvironment{proof}{\qquad \textit{Proof:}}{\hfill$\square$}
\makeatletter
\renewcommand*{\@opargbegintheorem}[3]{\trivlist
      \item[\hskip \labelsep{\emph{ #1\ #2}}] \emph{(#3):}\ \itshape}
\makeatother

\usepackage{tcolorbox}
\definecolor{mycolor}{rgb}{0.122, 0.435, 0.698}
\makeatletter
\newcommand{\mybox}[1]{%
  \setbox0=\hbox{#1}%
  \setlength{\@tempdima}{\dimexpr\wd0+13pt}%
  \begin{tcolorbox}[colframe=mycolor,boxrule=0.5pt,arc=4pt,
      left=6pt,right=6pt,top=6pt,bottom=6pt,boxsep=0pt,width=\@tempdima]
    #1
  \end{tcolorbox}
}
\makeatother

\usepackage{cite}

\begin{document}
\maketitle
\begin{abstract}
This paper investigates an infra-structure free  global localization of a group of communicating mobile agents (e.g., first responders or exploring robots) via an ultra-wideband (UWB) inter-agent ranging aided dead-reckoning.
We propose a loosely coupled cooperative localization  algorithm that acts as an augmentation atop the local dead-reckoning system of each mobile agent. This augmentation becomes active only when an agent wants to process a relative measurement it has taken.
The main contribution of this paper is addressing the challenges in the proper processing of the UWB range measurements in the framework of a loosely coupled cooperative localization. 
Even though UWB offers a decimeter level accuracy in line-of-sight (LoS) ranging, its accuracy degrades significantly in non-line-of-sight (NLoS) due to the significant unknown positive bias in the measurements. Thus, the measurement models for the UWB LoS and NLoS ranging conditions are different, and proper processing of NLoS measurements requires a bias compensation measure. We also show that, in practice, the measurement modal discriminators determine the type of UWB range measurements should be probabilistic. To take into account the probabilistic nature of the NLoS identifiers when processing UWB inter-agent ranging feedback, 
we employ an interacting multiple model (IMM) estimator in our localization filter. We also propose a bias compensation method for NLoS UWB measurements. The effectiveness of our cooperative localization is demonstrated via an experiment for a group of pedestrians who use UWB relative range measurements among themselves to improve their shoe-mounted INS geolocation.
\end{abstract}
 
\section{Introduction}
This paper investigates a practical infra-structure free solution for mobile asset (e.g., first responders or exploring robots) geo-localization in harsh indoor environments via cooperative localization using  ultra-wideband (UWB) inter-agent ranging aided dead-reckoning. In indoor localization, Global Positioning System (GPS) fails to provide accurate localization information due to obstructed line of sight to satellites and weak signal strength.
 Localization based on inertial navigation system (INS)~\cite{DT-JLW:04} or odometry~\cite{thrun2005probabilistic} provide a self-contained solution but suffer from unbounded error accumulation of inherent measurement noises overtime. Aiding by detecting and processing measurements from external landmarks helps to bound the localization error~\cite{JL-HFD:91,MWMGD-PN-SC-HFDW-MC:01} but has limited usage when external landmarks are not widely available. For a group of communicating mobile agents, aiding via cooperative localization (CL) by processing inter-agent measurements as feedback to update the location estimate makes the localization system more reliable under the circumstances when external landmarks are sparse~\cite{SSK-SF-SM:16}.
In CL, sporadic access to absolute external aiding signals like GPS by a particular member can result in a net benefit for the rest of the team when others take relative measurement from that particular agent. However, the effectiveness of CL depends 
on the accurate modeling and processing of the inter-agent measurements. On the other hand, in CL, the inter-agent measurement updates create strong correlations between agents. Ignoring the correlations will lead to over-confident estimations and even filter divergence. Keeping track of the correlation explicitly requires persistent inter-agent communication, therefore comes with high communication overhead and stringent connectivity requirement~\cite{SSK-SF-SM:16}.  Therefore, the effectiveness of CL also depends on devising consistent decentralized implantation that accounts for inter-agents correlations with reasonable computation and communication cost per agent. From the communication cost perspective, loosely coupled CL algorithms~\cite{POA-CR-RKM:01,HL-FN:13,DM-NO-VC:13,JZ-SSK-TRO:19},
 which account for unknown inter-agent correlations by implicit approaches that are closely related to the covariance intersection method in sensor fusion literature~\cite{SJJ-JKU:97},
 offer the most efficient solution. These algorithms do not require any network-wide connectivity; only the two agents involved in a relative measurement should exchange information to process that relative measurement. In this paper, we adopt the \emph{Discorrelated minimum variance} (DMV) approach of~\cite{JZ-SSK-TRO:19} as our CL~framework.

A variety of sensing technologies including computer vision-based techniques~\cite{JK-HJ:08} and wireless radio signal based techniques~\cite{RX-SVW-CR-NW:17} are used for inter-agent relative measurements in CL implementations. The computer vision-based techniques' requirement of LoS condition between the agents and proper lighting make them less effective in complex and cluttered environments. For such environments, wireless signal based inter-agent rangings offer a more efficient solution. Among the wireless ranging technologies, UWB  due to its high time resolution, wide bandwidth, and capability to work under NLoS condition~\cite{LY-GBG:04} has attracted significant attention for applications in dense multi-path environments, especially indoor environments. UWB uses a time-of-flight approach for ranging and offers a decimeter level ranging in LoS conditions~\cite{FZ-AG-KKL:19}. However, NLoS UWB ranging measurements are positively biased (see Fig.~\ref{fig:NLoS_bias}) and thus have lower accuracy~\cite{BD-JK-ND:03}, which can have a significant impact on localization performance. Thus, the measurement models for the UWB LoS and NLoS ranging conditions are different, and proper processing of NLoS measurements requires a bias compensation measure.

 \begin{figure}[t]
    \centering
    \includegraphics[width=0.15\textwidth]{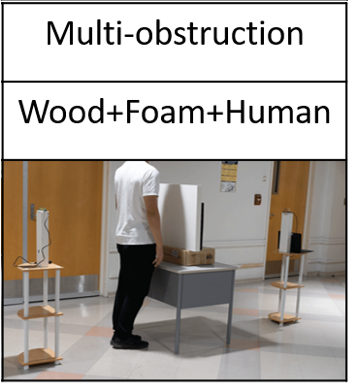}\includegraphics[width=0.37\textwidth]{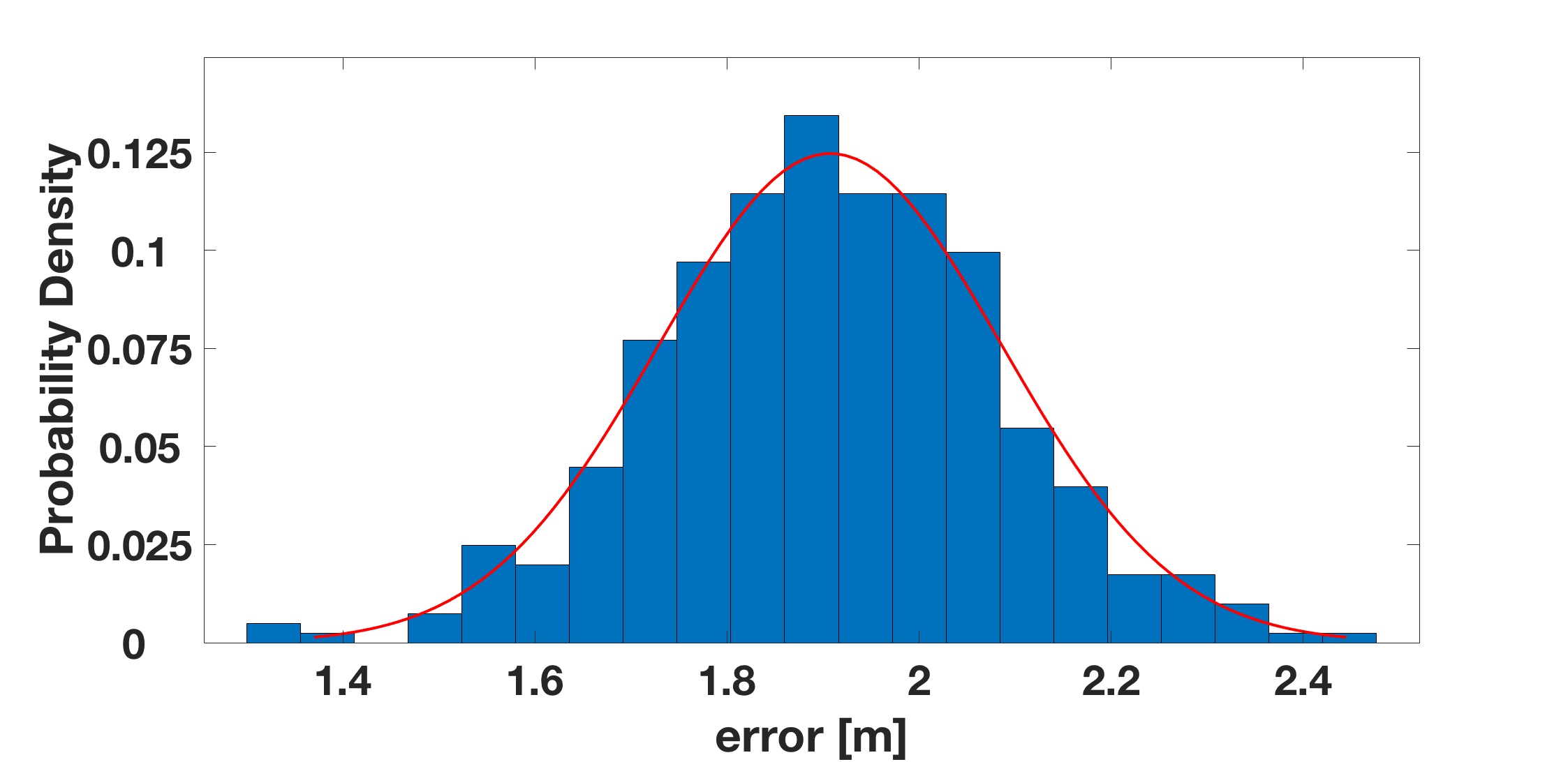}
    \caption{{\small Experimental study with multiple obstruction: the plot on the right shows the error probability distribution when the actual distance is $10$ meters; see~\cite{JZ-SSK:19sensor} for more case studies.
    }}
    \label{fig:NLoS_bias}
\end{figure}
To mitigate the adverse effect of the NLoS ranging bias on the localization accuracy, one idea is to identify the NLoS measurements and drop them~\cite{ZZ-SL-LW:18,NR-PL-NP-SC-BS-AR:18,MX-CX-DY-LW:18}. But, this approach limits the effectiveness of the UWB measurement feedbacks in dense and complex environments. To avoid discarding the NLoS measurements, empirical analysis and machine learning methods that aim to identify and remove the bias are proposed ~\cite{BA-KP:06,SW-YM-QZ-NZ:07,SM-WMG-HW-MZW:10,QZ-DZ-SZ-TZ-DM:15,KW-KY-YL:17}. However, these approaches require collecting a large amount of training data. The machine learning techniques also come with high computational complexity to analyze the signal channel statistics. As such, these methods are not a practical solution for real-time online applications in unknown environments. When UWB is used as an aiding for a dead-reckoning system,~\cite{JZ-SSK-ITM:19,JZ-SSK:19sensor,JZ-SSK:20PLANs} use the algorithmic bias compensation methods from estimation filter literature~\cite{YBS-PKW-XT:11} to deal with UWB NLoS bias. \cite{JZ-SSK-ITM:19} uses the covariance inflation method followed by a constrained Kalman filtering to compensate for bias in UWB range measurements in a cooperative localization algorithm. However, the covariance inflation method is known to be conservative and can lead to filter inconsistency~\cite{YBS-PKW-XT:11}. On the other hand,~\cite{JZ-SSK:19sensor} and \cite{JZ-SSK:20PLANs} use the Schmidt Kalman filtering (SKF)~\cite{RYN-SMH-SMG-ABP:05,YBS-PKW-XT:11}, which is known to yield a more efficient bias compensation, followed by a novel constrained sigma point based filtering to process NLoS measurements with respect to beacons with known locations to aid an INS localization. In this paper, we adopt the SKF bias compensation approach for NLoS UWB bias compensation and incorporate it into the framework of the DMV CL.

Localization filters that process both LoS and NLoS UWB ranging such as those in~\cite{JZ-SSK-ITM:19} and \cite{JZ-SSK:19sensor} assume that the LoS and NLoS measurements can be identified and distinguished from each other with exact certainty. These localization filters use the popular power-based NLoS identification method~\cite{KG-AKR-YS-CLL-GC:17}. The working principle of the power-based NLoS identification methods is that in LoS condition power of the received direct-path signal takes a big proportion of the total received signal power, while in NLoS condition the direct-path is significantly attenuated or even completely blocked. When the difference between total received power and the direct-path power is larger than a threshold value, the range measurement is identified to be NLoS~\cite{KG-AKR-YS-CLL-GC:17}.  The performance of this approach however depends highly on the choice of the discrimination threshold value. Moreover, as we demonstrate via an experimental study in our preliminary work~\cite{JZ-SSK:20PLANs}, in practice deterministic identification of the UWB ranging mode is not accurate, and identification that determines the type of UWB range measurements deliver their results with only some level of certainty, see Fig.~\ref{fig:exp_identification}.  Given  the  probabilistic  nature  of  the  power-based  LoS/NLoS identification method,  processing inter-agent UWB  range  measurements should be  modeled as  a \emph{dynamic multiple model problem}. Optimal estimation of a dynamic multiple model problem requires a set of parallel filters whose number increases exponentially with time~\cite{YBS-XRL:95,XL-YB:96}. To design a practical localization algorithm with a reasonable computation cost, we adopt the suboptimal IMM estimator framework~\cite{YB-PKW-XT:11}.

To summarize, the main contribution of this paper is to propose a proper framework to process multi-modal UWB range measurements, which are multi-modal due to the possibility of the measurements being in LoS or NLoS. 
The innovation in our work is to take into account the probabilistic nature of the NLoS identifiers and also propose a bias compensation method for NLoS measurements for an UWB-based cooperative localization in complex environments. The effectiveness of our proposed method is demonstrated via an experiment for a group of pedestrians who use UWB relative range measurements among themselves to improve their shoe-mounted INS geolocation. We note that UWB enables also inter-agent communication for cooperative localization~\cite{JZ-SSK:20sensor}. Therefore, our solution provides an infra-structure free localization to track assets in challenging indoor venues where the environment is not fixed to offer features for SLAM~\cite{MWMGD-PN-SC-HFDW-MC:01} (e.g., indoor firegrounds), the lighting is poor (e.g. underground cave) or the features are not revisited (e.g. in rapidly evolving fire scenes).

 \begin{figure}[t]
    \centering
    \includegraphics[scale=0.26
    ]{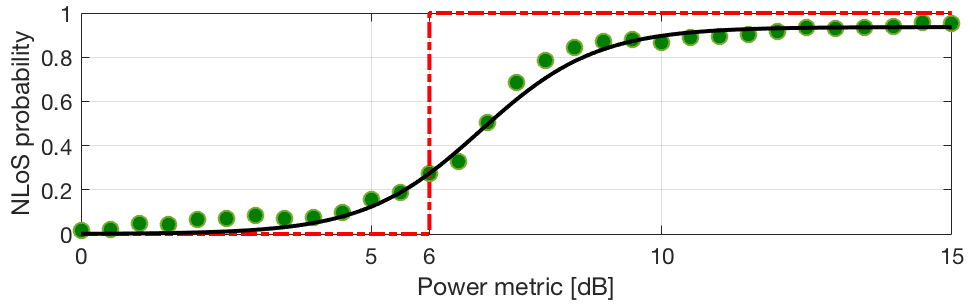}
    \caption{{\small The experimental result that demonstrates the probabilistic nature of the power-based NLoS identification: the green points are the empirical probability of a signal being in NLoS with the given power metric, and the black curve is the fitted sigmoid probability function $p=\frac{1}{1.068+1.013\textup{e}^{(-PM+6.934)}}$, where $PM$ is the power metric. The dashed vertical line shows the conventional deterministic threshold that if used, to its left corresponds to identifying the signal as LoS and to its right as NLoS with absolute certainty.   
    }}
    \label{fig:exp_identification}
\end{figure}

\section{Problem definition}\label{sec::prob_def}
Consider a team of $N$ mobile agents with computation and communication capabilities in which each agent is equipped with a set of proprioceptive sensors, e.g., INS or wheel encoders, to localize itself in a global frame using dead-reckoning. At each time step $t\in\mathbb{Z}_{\geq0}$ let this local state estimate and the corresponding error covariance constitute the local belief, $\text{bel}^{i\prpg}(t)=(\Hvect{x}^{i\prpg}(t),\vect{P}^{i\prpg}(t))$, of agent $i$ about its pose in the global frame. To bound the localization error of the dead-reckoning system, suppose each agent is also equipped with an UWB transceiver to take and process relative range measurements from other team members as well as possibly a few UWB beacons with known locations in the environment. We assume that agents are able to uniquely identify the other UWB nodes (hereafter, UWB nodes refers to both a mobile agent or a beacon) via the unique MAC address of their UWB transceivers. The measurement model for the UWB ranging between any agent $i$ and another node $j$ is 
\begin{align}\label{eq::measur_ij_bias}
  {z}^i_{j}(t)&=\underbrace{\|\vect{p}^i(t)-\vect{p}^j(t)\|}_{h(\vect{x}^i(t),\vect{x}^j(t))}+\,{b}^i(t)+{\nu}^i(t)+\begin{cases} 
      0, & \text{LoS} \\
      {b}^i(t), & \text{NLoS}
\end{cases},
\end{align}
where ${b}^i(t)$ is the additive bias modeled as Gaussian noise with mean $\bar{\phi}^i$ and variance $E[b^i(t)\,b^i(t)]={\Phi}^i>0$, while ${\nu}^i(t)$ is the additive zero-mean white Gaussian measurement noise with variance $E[\nu^i(t)\,\nu^i(t)]={R}^i>\vect{0}$. NLoS signal propagations can be distinguished from the LoS signal propagation based on a real-time signal power-based approach without any prior information about the environment~\cite{KG-AKR-YS-CLL-GC:17}. The Previous work in~\cite{JZ-SSK:19sensor} assumed that using a power-based identification method we can distinguish LoS and NLoS ranging conditions from each other with exact certainty by implementing a separation threshold. Then, the belief updates using relative inter-agent measurement processing can be carried out using the respective measurement model. However, in practice, the identification methods do not exactly identify the measurement condition with absolute certainty. As we have shown in our preliminary work~\cite{JZ-SSK:20PLANs}, if we record the power metric of our power-based modal identifier under a controlled environment where we know the true measurement type, we arrive at a probabilistic identification outcome as shown in Fig.~\ref{fig:exp_identification}. Therefore, what the power-based UWB mode discriminator is delivering is a likeliness level about the measurement mode.
Let $M(t)\in\{M_1,M_2\}$ be the modal state of the UWB ranging measurement at time $t$. The power-based UWB mode discriminator assigns a normalized probability that the measurement is in LoS (denoted by $M_1$) or NLoS (denoted by $M_2$),  
\begin{align}
    \label{eq::prob_model}
    p(M(t)=M_i),~~ i\in\{1,2\}, ~\text{where}~p(M_1)\!+\!p(M_2)=1.
\end{align}
As our experiment in Section~\ref{sec::exp_eval} shows, using a threshold in the power-based UWB mode discriminator to assign a deterministic measurement type, and then a consequent measurement processing leads to an inferior localization result.  

On the other hand, a relative measurement update (in LoS or NLoS) using a feedback gain $\vect{K}^i$ in the form of
 \begin{align}\label{eq::our-x_c}
    \Hvect{x}^{i\updt}&=\Hvect{x}^{i\prpg}+\vect{K}^i\,({z}^i_{j}-\hat{z}^i_{j}),
\end{align}
creates a correlation among the state estimates of agents $i$ and $j$, i.e., $\vect{P}_{ij}^{\updt}\neq \vect{0}$ after implementing~\eqref{eq::our-x_c}. To maintain the exact account of the updated and the propagated cross-covaraince terms for filter consistency, agents need to communicate with each other at all times. However, under limited connectivity condition, it is ideal that 
 agent $i$ and agent $j$  communicate if and only if a relative measurement is taken between them. Our objective in this paper is to design a relative measurement processing method that respects this minimal communication connectivity requirement while also takes into account the stochastic nature of the UWB ranging measurement modal variable $M(t)$.

\section{An IMM based estimator with UWB ranging feedback}
\label{sec::imm_cl}
In this section, we derive the constituting equations of the IMM filtering for estimate correction via UWB ranging feedback. To simplify the notation, we derive our equations for when an agent $i$ takes only one measurement with respect to other nodes at each time. The case of multiple concurrent measurements is discussed in the next section, when we implement our IMM based estimator in the context of the DMV based loosely coupled CL.

Note that by implementing a power-based UWB modal discriminator~\cite{KG-AKR-YS-CLL-GC:17}, a confidence level about the measurement mode can be derived with probability density function
\begin{align}
    \label{eq::prob_model_1}
    f_M(m)=\sum_{n=1}^2 p_n(t)\delta(M_m-M_n),\quad m\in\{1,2\},
\end{align}
where $p_n(t)\in[0,1]$ with $\sum_{n=1}^2p_n(t)=1$, and $\delta$ is the Dirac measure. Here the subscript $1$ represents the unbiased LoS ranging mode and the subscript $2$ represents the biased NLoS ranging mode. Note that the density function~\eqref{eq::prob_model_1} is independent of the modal history, thus for any UWB ranging mode $n\in\{1,2\}$, we can always write
\begin{align}\label{eq::modal_indep}
{P}(M(t)=&M_n|M(t-1)=M_m)\nonumber\\&=\!P(M(t)=M_n)\!=\!p_n(t), ~m\!\in\!\{1,2\}.
\end{align}
Next, let the aggregate exteroceptive measurements history taken by agent $i$ from initial time to time step $t$ be $\vect{Z}^i_{1:t}$. Moreover, let the $l$th model hypotheses sequence, through time $t$ be 
\begin{align}\label{eq::measure_history}
    M^i_{t,l}=\{M_{m_{1,l}},M_{m_{2,l}},...,M_{m_{t,l}}\},
\end{align}
where $m_{t,l}\in\{1,2\}$ is the measurement model index at time $t$. Because for each time step, there are two possible measurement models, then $2^t$ different measurement model hypotheses sequences exist at time $t$, i.e., $l\in\{1,...,2^t\}$. The conditional probability density function of the state $\vect{x}^i(t)$ at time step $t$ is obtained using the total probability theorem with respect to the mutually exclusive and exhaustive set of events~\eqref{eq::measure_history} with $l\in\{1,...,2^t\}$, as a Gaussian mixture with an exponentially increasing number of terms
\begin{align}\label{eq::optimal_combine}
    p(\vect{x}^i(t)|\vect{Z}^i_{1:t})=\sum_{l=1}^{2^t}p(\vect{x}^i(t)|M^i_{t,l},\vect{Z}^i_{1:t})P(M^i_{t,l}|\vect{Z}^i_{1:t}),
\end{align}
where $p(\vect{x}^i(t)|M^i_{t,l},\vect{Z}^i_{1:t})$ is the model-conditioned updated distribution and $P(M^i_{t,l}|\vect{Z}^i_{1:t})$ is the probability of the $l$th model hypotheses sequence conditioned on the observations. From~\eqref{eq::optimal_combine}, $2^t$ filters are needed to run in parallel to derive the exact distribution. The computational and memory complexity makes the optimal method impractical. IMM estimator~\cite[chapter 1]{YB-PKW-XT:11} is a feasible sub-optimal solution that only requires the number of filters linear to the number of models operating in parallel in each step. Following~\cite{HAPB-YBS:88}, a cycle of the IMM estimator from right after the previous measurement update up to and including the current measurement update includes the following steps:
\begin{itemize}
    \item Mixing:
    \begin{subequations}\label{eq::imm_estimator_m}
\begin{align}
&\!\!P(M(t)|\vect{Z}^i_{1:t-1})\!\leftarrow\!P(M(t\!-\!1)|\vect{Z}^i_{1:t\!-\!1}) ,\label{eq::m_1}\\
&\!\!p(\vect{x}^i(t\!-\!1)|M(t),\vect{Z}^i_{1:t-1})\!\leftarrow\! p(\vect{x}^i(t\!-\!1)|M(t\!-\!1),\vect{Z}^i_{1:t-1}),\label{eq::m_2}
\end{align}
\end{subequations}
    \item Model-based propagation:
\begin{align}\label{eq::imm_estimator_p}
p(\vect{x}^i(t)|M(t),\vect{Z}^i_{1:t-1})&\leftarrow p(\vect{x}^i(t-1)|M(t),\vect{Z}^i_{1:t-1}),
\end{align}
\item Probability evolution:
\begin{align}\label{eq::imm_estimator_e}
P(M(t)|\vect{Z}^i_{1:t})&\leftarrow P(M(t)|\vect{Z}^i_{1:t-1}),
\end{align}
\item Model-based updating:
\begin{align}\label{eq::imm_estimator_c}
  &p(\vect{x}^i(t)|M(t),\vect{Z}^i_{1:t})\leftarrow p(\vect{x}^i(t)|M(t),\vect{Z}^i_{1:t-1}),
\end{align}
\item Combination:
\begin{align}
    p(\vect{x}^i(t)|\vect{Z}^i_{1:t})&\leftarrow p(\vect{x}^i(t)|M(t),\vect{Z}^i_{1:t}) .\label{eq::imm_combine_1}
\end{align}
\end{itemize}
In IMM estimator, the cycle is initialized from the model-conditioned updated distributions $p(\vect{x}^i(t-1)|M(t-1),\vect{Z}^i_{1:t-1})$ and the model probability based on the observation history $P(M(t-1)|\vect{Z}^i_{1:t-1})$ from the previous cycle. To simplify the notation, we use $M_n(t-1)$ to represent $M(t-1)=M_n,n\in\{1,2\}$. \eqref{eq::m_1} is expanded according to the Chapman-Kolmogorov equation~\cite{AP-SUP:02} as
\begin{align}
    &P(M_n(t)|\vect{Z}^i_{1:t-1})\nonumber\\&=\sum_{m=1}^2P(M_n(t)|M_m(t-1))P(M_m(t-1)|\vect{Z}^i_{1:t-1}).
    \label{eq::imm_mixing_prop}
\end{align}
Given~\eqref{eq::modal_indep} and $\sum_{m=1}^2P(M_m(t-1)|\vect{Z}^i_{1:t-1})=1$, however,~\eqref{eq::imm_mixing_prop} results in 
\begin{align}\label{eq::mix_1}P(M_n(t)|\vect{Z}^i_{1:t-1})=p_n(t),\quad n\in\{1,2\}.\end{align}

Next, note that based on the law of total probability,~\eqref{eq::m_2} reads~as
\begin{align}
    &\!\!\!\!p(\vect{x}^i(t\!-\!1)|M_n(t),\vect{Z}^i_{1:t-1})\!=\!\!\!\!\sum_{m=1}^{2}\!p(\vect{x}^i(t\!-\!1)|M_m(t\!-\!1),\vect{Z}^i_{1:t\!-\!1})\nonumber\\&\quad\qquad \times P(M_m(t\!-\!1)|M_n(t),\vect{Z}^i_{1:t-1}), ~~n\in\{1,2\}.
    \label{eq::imm_initialization}
\end{align}
However, as Lemma~\ref{lemma::01} shows, by invoking~\eqref{eq::modal_indep},~\eqref{eq::imm_initialization} can be simplified to~\eqref{eq::first_lem}.

\begin{lem}
\label{lemma::01}
Given the probability density function~\eqref{eq::prob_model_1} model for UWB ranging mode type and~\eqref{eq::modal_indep}, then we have
\begin{align}
p(\vect{x}^i(t - 1)|M_n(t),\vect{Z}^i_{1:t-1})&=p(\vect{x}^i(t - 1)|\vect{Z}^i_{1:t-1}),\label{eq::first_lem}
\end{align}
for $n\in\{1,2\}$.
\end{lem}
\begin{proof}
Note that by virtue of Bayes rule, we obtain
\begin{align}
    &\!\!\!P(M_m(t-1)|M_n(t),\vect{Z}^i_{1:t-1})\nonumber\\&=\frac{P(M_n(t)|M_m(t-1),\vect{Z}^i_{1:t-1})P(M_m(t-1)|\vect{Z}^i_{1:t-1})}{\sum_{m=1}^{2}P(M_n(t)|M_m(t-1))P(M_m(t-1)|\vect{Z}^i_{1:t-1})}.
    \label{eq::imm_initialization_prop}
\end{align}
By virtue of~\eqref{eq::modal_indep},~\eqref{eq::imm_initialization_prop} can be written as
\begin{align*}
       &P(M_m(t-1)|M_n(t),\vect{Z}^i_{1:t-1})\nonumber\\&\qquad=\frac{p_n(t)P(M_m(t-1)|\vect{Z}^i_{1:t-1})}{\sum_{m=1}^{2}p_n(t)P(M_m(t-1)|\vect{Z}^i_{1:t-1})}\\&\qquad=P(M_m(t-1)|\vect{Z}^i_{1:t-1}),
\end{align*}
for $n\in\{1,2\}$, and $m\in\{1,2\}$. Here, we also used $\sum_{m=1}^{2}P(M_m(t-1)|\vect{Z}^i_{1:t-1})=1$. Then~\eqref{eq::imm_initialization}, for $n\in\{1,2\}$, is equivalent to
\begin{align*}
        &p(\vect{x}^i(t-1)|M_n(t),\vect{Z}^i_{1:t-1})\\=&\sum_{m=1}^{2}p(\vect{x}^i(t-1)|M_m(t-1),\vect{Z}^i_{1:t-1})P(M_m(t-1)|\vect{Z}^i_{1:t-1}).
\end{align*}
From~\eqref{eq::imm_combine} we can write 
\begin{align*}
p(\vect{x}^i(t-1)|M_n(t),\vect{Z}^i_{1:t-1})=p(\vect{x}^i(t-1)|\vect{Z}^i_{1:t-1}),
\end{align*}
for $n\in\{1,2\}$, which concludes the proof.
\end{proof}

Lemma~\ref{lemma::01} states that~\eqref{eq::m_2} of mixing step is not needed and the model-conditioned posterior distribution $p(\vect{x}^i(t - 1)|M_n(t),\vect{Z}^i_{1:t-1})$ can be determined directly from the posterior distribution $p(\vect{x}^i(t - 1)|\vect{Z}^i_{1:t-1})$ of the previous cycle. Then, since the propagation model of agent $i$ does not have dependency on modal state, the model-based propagation~\eqref{eq::imm_estimator_p} then is simply obtained from propagating the posterior distribution from previous step through the system model, i.e., 
$$
p(\vect{x}^i(t)|\vect{Z}^i_{1:t-1})\leftarrow p(\vect{x}^i(t-1)|\vect{Z}^i_{1:t-1}),
$$
i.e., 
In summary, the flow shown in Fig.~\ref{fig::flow_imm_origin} without mixing step~\eqref{eq::imm_estimator_m} and model-based propagation step~\eqref{eq::imm_estimator_p} will be equivalent to traditional IMM estimator given the UWB ranging mode type probability density function~\eqref{eq::prob_model_1}. This property simplifies the IMM estimator and makes implementation of the IMM CL easier as an augmentation service atop the local filters. 
\begin{figure}
    \centering
    \includegraphics[scale=0.34]{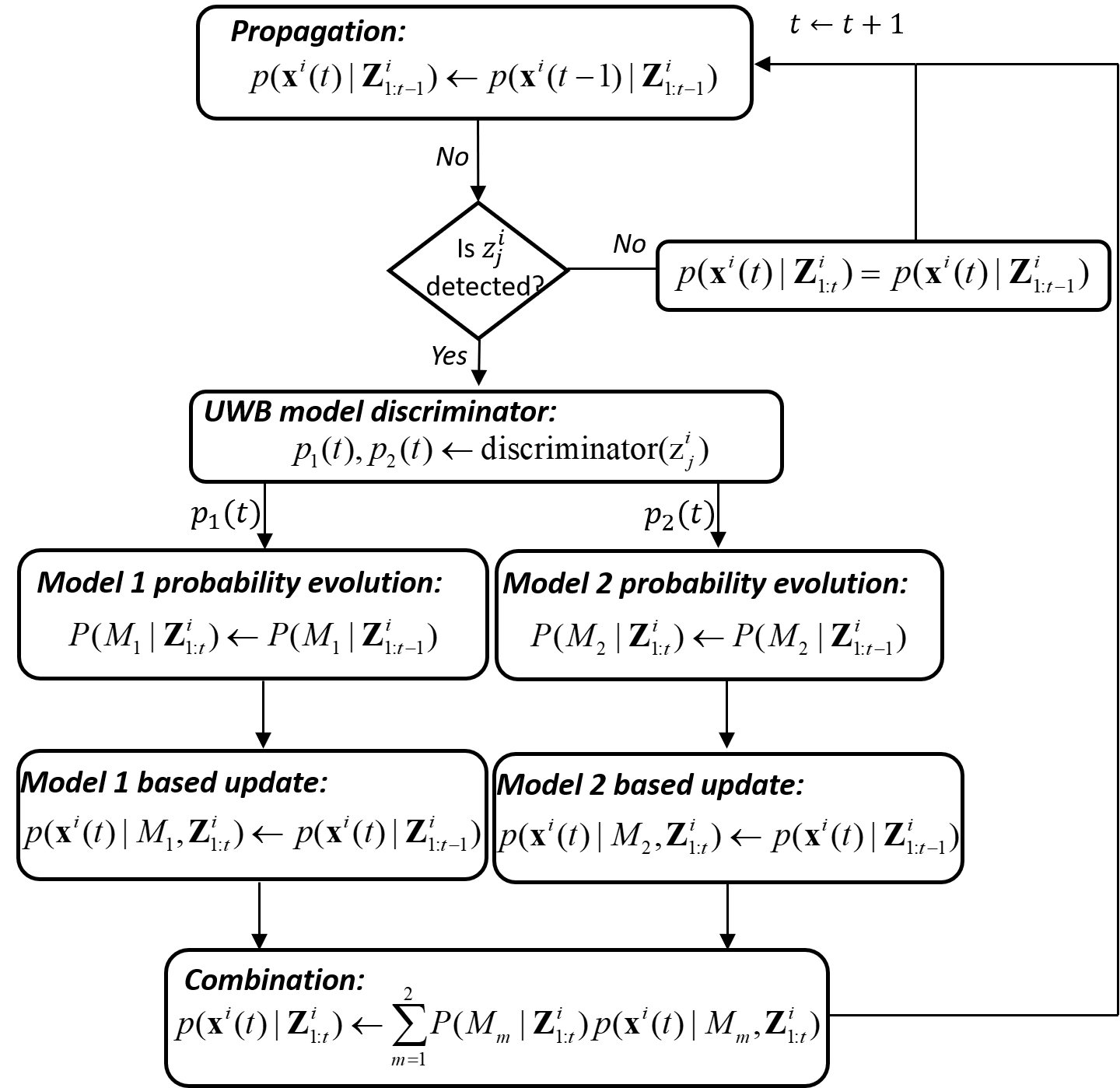}
    \caption{{\small One cycle of IMM estimator of agent $i$ with UWB ranging correction feedback.}}
    \label{fig::flow_imm_origin}
\end{figure}
Next, we note that the propagated distribution is updated in two parallel process conditioned on different measurement models as in Fig.~\ref{fig::flow_imm_origin} to derive the model-conditioned updated distribution $p(\vect{x}^i(t)|M_n(t),\vect{Z}^i_{1:t})$. For $n\in\{1,2\}$, the model probability is evolved according to 
\begin{align}
  \!\!  P(M_n(t)|\vect{Z}^i_{1:t})\!=\!\!\frac{P(z_j^i(t)|M_n(t),\vect{Z}^i_{1:t-1})P(M_n(t)|\vect{Z}^i_{1:t-1})}{\sum\limits_{m=1}^{2}\!\!\!P(z_j^i(t)|M_m(t),\vect{Z}^i_{1:t-1})P(M_m(t)|\vect{Z}^i_{1:t-1})}
    \label{eq::imm_evolve}
\end{align}
$P(z_j^i(t)|M_n(t),\vect{Z}^i_{1:t-1})$ is the model-conditioned likelihood, which can be derived from the likelihood function of the model $M(t)$ if the distribution is Gaussian as follows~\cite[chapter 2]{YBS-XRL-TK:01}  
\begin{align}
  p(z^{i}_{j}(t)|M_n(t),\vect{Z}^{i}_{1:t})=\frac{\textup{e}^{(-{\tilde{z}_{j_n}^{i}}{}^2/2{S_{j_n}^i})}}{\sqrt{2\pi|{S_{j_n}^i}|}},  
\end{align}
where $\tilde{z}_{j_n}^{i}=z_{j}^{i}-\hat{z}_{j_n}^{i}$ and $S_{j_n}^i$ are the model-matched innovation and corresponding covariance. In IMM approach, the optimal estimate~\eqref{eq::optimal_combine} is approximated finally by
\begin{align}
    p(\vect{x}^i(t)|\vect{Z}^i_{1:t})=\sum_{n=1}^{2}P(M_n(t)|\vect{Z}^i_{1:t})p(\vect{x}^i(t)|M_n(t),\vect{Z}^i_{1:t}).
    \label{eq::imm_combine}
\end{align}


\begin{figure*}
    \centering
    \includegraphics[scale=0.2]{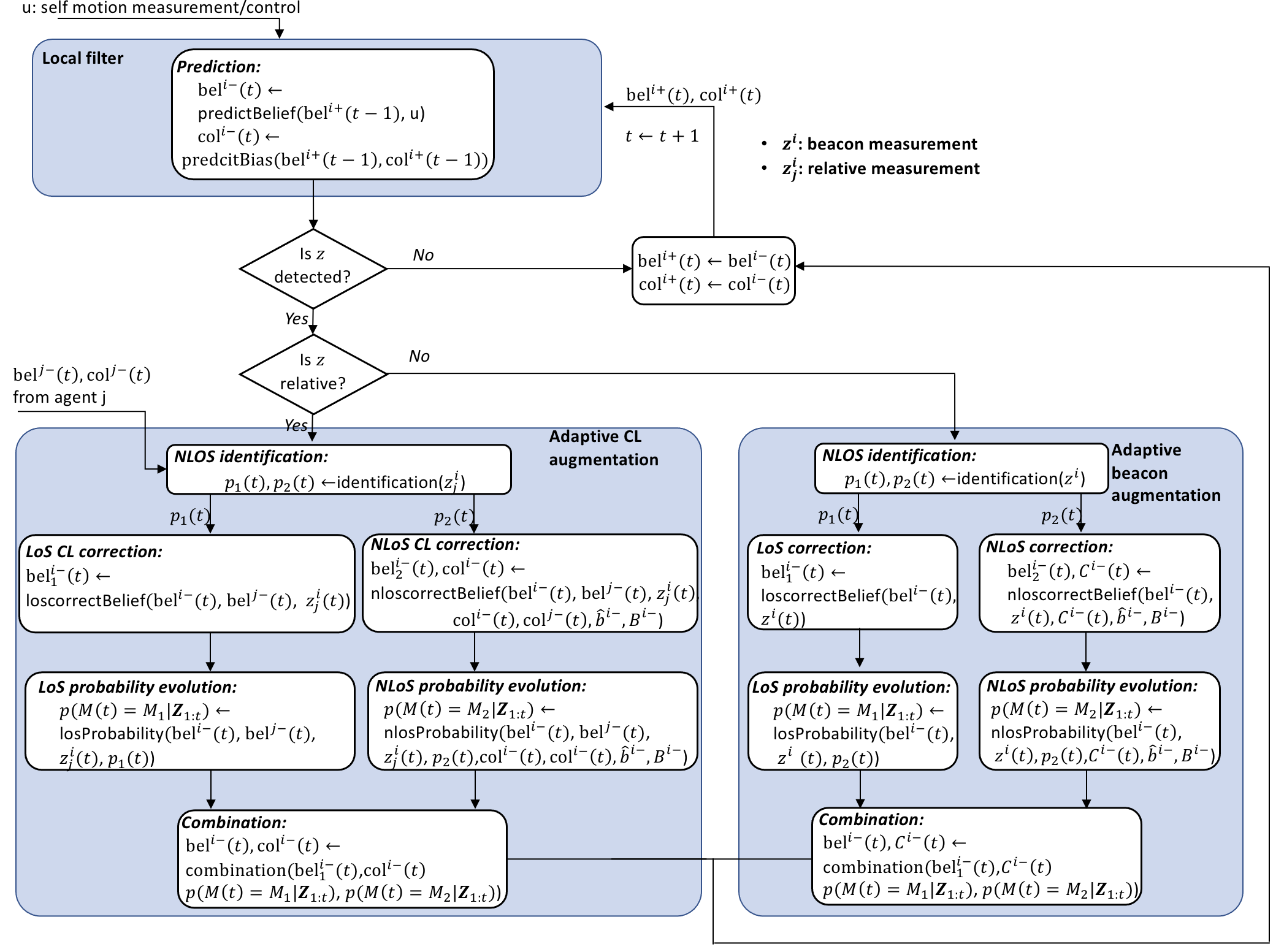}
    \caption{{\small The proposed AUCL, an augmentation atop of the local filter of agent $i$ becomes active when there is an inter-agent UWB range measurement. It contains two parallel updating filter, one is used to process unbiased LoS measurement and the other one is used to process NLoS measurement with bias compensation. }}
    \label{fig:flow_cl}
\end{figure*}
\section{An IMM based cooperative localization via UWB inter-agent ranging}
Given the IMM estimator in Fig.~\ref{fig::flow_imm_origin} for the UWB ranging correction feedback, we propose the adaptive UWB-based cooperative localization (AUCL) algorithm shown in Fig.~\ref{fig:flow_cl} to process the UWB-based relative range measurements taken by agent $i$ from another node $j$ in the form of a loosely coupled augmentation.  To develop our loosely coupled CL we employ the DMV approach of~\cite{JZ-SSK-TRO:19}. For notational simplicity, our algorithm is depicted for when there is a single relative measurement taken by agent $i$ at each time. To process multiple concurrent relative measurements, we use sequential updating (see~\mbox{\cite[page 103]{YB-PKW-XT:11}}). That is, agent $i$ first collects the local belief of the agents that it has taken relative measurements from at time $t$. Then, it processes them via our proposed methods one after the other by using its previously updated belief as its local belief. In what follows, we explain the components of the AUCL algorithm. 

\textbf{\textit{predictBelief} function ( $p(\vect{x}^i(t)|\vect{Z}^i_{1:t-1}))\leftarrow p(\vect{x}^i(t-1)|\vect{Z}^i_{1:t-1})$):}  At each time step $t\in\mathbb{Z}^{+}$, the dead-reckoning system (e.g., INS or odometery) of each agent $i$ propagates an estimate of the ego state $\Hvect{x}^{i\prpg}(t)=\mathsf{f}^i(\Hvect{x}^{i\updt}(t-1),\vect{u}^i(t))\in\real^{n_x}$ and the corresponding positive definite error covariance matrix $\vect{P}^{i\prpg}(t)\in\mathbb{S}^{++}_{n_x}$, in a global frame (e.g., the global earth-fixed coordinate frame with axes pointing north, east and down for an INS system). This dead-reckoning process is executed through the \textit{predictBelief} function in the AUCL algorithm. When there is no exteroceptive measurement to update the local belief, we set $\text{bel}^{i\updt}(t)=\text{bel}^{i\prpg}(t)=(\Hvect{x}^{i\prpg}(t),\vect{P}^{i\prpg}(t))$, otherwise we proceed to correct the belief as outlined below.

 \textbf{\textit{loscorrectBelief} function} $(p(\vect{x}^i(t)|M_1(t),\vect{Z}^i_{1:t})\leftarrow p(\vect{x}^i(t)|\vect{Z}^i_{1:t-1})$): Let the relative range measurement $z^i_j(t)$ by agent $i$ from any \emph{mobile} agent $j$ be in LoS. Since there is no bias in the measurement, to correct the local belief of agent $i$ using this measurement (\textit{loscorrectBelief} function in Fig.~\ref{fig:flow_cl}), we employ the DMV update. The idea in DMV approach is that instead of maintaining the cross-covariance term $\vect{P}^{\prpg}_{ij}$ in the joint covariance matrix of any two agents $i$ and $j$, we use the conservative upper bound below 
\!\!\!\!\begin{align}\label{eq::PJ-BPJ}
\!\!\begin{bmatrix}
 \vect{P}^{i\prpg}(t)&\!\!\vect{P}_{ij}^{\prpg}(t)\\
 {\vect{P}_{ij}^{\prpg}}(t)^\top&\!\!\vect{P}^{j\prpg}(t)
 \end{bmatrix}\!\leq \begin{bmatrix}\frac{1}{\omega}\vect{P}^{i\prpg}(t)&\vect{0}\\\vect{0}&\!\!\frac{1}{1-\omega}\vect{P}^{j\prpg}(t)
\end{bmatrix}
\!,~\omega\!\in\![0,1],
\end{align}
to obtain $\Bvect{\mathsf{P}}^{i}(\omega)$ that satisfies $\mathrm{E}_{\text{f}}[(\vect{x}^i-\Hvect{x}^{i\updt})(\vect{x}^i-\Hvect{x}^{i\updt})^\top]\leq\Bvect{\mathsf{P}}^{i}(\omega)$, and has no dependency on $\vect{P}_{ij}^{\prpg}(t)$.
Then, a `minimum variance' like update gain $\vect{K}^i$ in~\eqref{eq::our-x_c} is obtained from minimizing the trace of $\Bvect{\mathsf{P}}^{i}(\omega)$. Following~\cite{JZ-SSK-TRO:19}, the updated belief by processing LoS measurements $\text{bel}^{i\updt}_{\text{1}}(t)=(\Hvect{x}^{i\updt}_{\text{1}}(t),\vect{P}^{i\updt}_{\text{1}}(t))$ (subscript 1 is used to represent LoS condition for simplicity) for agent $i$ is ((\textit{loscorrectBelief}) function in Fig.~\ref{fig:flow_cl})
\begin{align*}
\Hvect{x}_1^{i\updt}&=\Hvect{x}^{i\prpg}_1+\Bvect{\mathsf{K}}_{1}^i(\omega_\star^i)\,({z}^i_{j}-\hat{z}^i_{j1}),~~\vect{P}^{i\updt}_1=\Bvect{\mathsf{P}}^{i}(\omega_\star^i),
\end{align*}
where 
\begin{align}\label{eq::Kl_omega}
    \Bvect{\mathsf{K}}^i(\omega)\!=\! \frac{\vect{P}^{i\prpg}}{\omega}{\vect{H}_i^i}^{\top}\big(\vect{H}_{i}^i\frac{\vect{P}^{i\prpg}}{\omega}{\vect{H}_{i}^i}^\top\!\!\!+\!\vect{H}_{j}^i\frac{\vect{P}^{j\prpg}}{1\!-\!\omega}{\vect{H}_{j}^i}\!^\top\!\!\!+\!{R}^i\big)^{-1}\!\!\!.
\end{align}
Using this gain that minimizes the trace of $\Bvect{\mathsf{P}}^{i}(\omega)$, we obtain
\begin{align}\label{eq::P_l_omega_inverse}
    \Bvect{\mathsf{P}}^{i}(\omega)=&
  \big(\omega(\vect{P}^{i\prpg})^{-1}+(1-\omega){\vect{H}_i^{i}}^{\top}(\vect{H}_j^{i}\vect{P}^{j\prpg}{\vect{H}_j^{i}}^{\top}\nonumber\\&~~~~~~~\qquad\qquad\qquad+(1-\omega){R}^{i})^{-1}\vect{H}_i^{i}\big)^{-1}\!.
\end{align} 
where the optimal $\omega^\star\in[0,1]$ is obtained from
\begin{align}\label{eq::omega_star}
\omega^i_\star=\underset{0\leq\omega\leq 1}{\argmin} ~\log\det\,&\Bvect{\mathsf{P}}^{i}(\omega),
\end{align}
where $\vect{H}^i_{i}\!=\!\partial {h}(\Hvect{x}^{i\prpg}\!,\Hvect{x}^{j\prpg})/\partial{\vect{x}}^i$ and $\vect{H}^i_{j}\!=\!\partial {h}(\Hvect{x}^{i\prpg}\!,\Hvect{x}^{j\prpg})/\partial{\vect{x}}^j$ are elements of the linearized model of $h(\Hvect{x}^{i\prpg},\Hvect{x}^{j\prpg})$.

\textbf{\textit{nloscorrectBelief} function $(p(\vect{x}^i(t)|M_2(t),\vect{Z}^i_{1:t})\leftarrow p(\vect{x}^i(t)|\vect{Z}^i_{1:t-1})$):} Let the relative range measurement $z^i_j(t)$ by agent $i$ from \emph{mobile} agent $j$ be in NLoS. To account for the measurement bias, recall~\eqref{eq::measur_ij_bias}, we use a SKF framework\cite{YBS-PKW-XT:11}. In SKF, the bias is appended to the states as a random variable, to account for its uncertainty and correlation with the state estimate but its value is not updated based on the measurement feedback. We let the joint extended state and the prior belief of agent $i$ and $j$ at time $t$ be, respectively,
$\vect{x}_J(t)\!=\!(\vect{x}^{i}(t)^\top\!, \vect{x}^{j}(t)^\top,b^{i\prpg}(t))^\top$, and $\text{bel}_J^{\prpg}(t)\!=\!(\Hvect{x}_J^{\prpg}(t),\vect{P}_J^{\prpg}(t))$~where
\begin{align*}
\Hvect{x}_J^{\prpg}(t)&\!=\!\!\begin{bmatrix}\Hvect{x}^{i\prpg}(t)\\ \Hvect{x}^{j\prpg}(t)\\\hat{b}^{i\prpg}(t)\end{bmatrix}\!\!,~
\vect{P}_{J}^{\prpg}(t)\!=\!\!\begin{bmatrix}\vect{P}^{i\prpg}(t)&\vect{P}_{ij}^{\prpg}(t)&\vect{C}^{ii\prpg}(t)\\
{\vect{P}_{ij}^{\prpg}}^{\top}(t)&\vect{P}^{j\prpg}(t)&\vect{C}^{ji\prpg}(t)\\{\vect{C}^{ii\prpg}}(t)^{\top}&{\vect{C}^{ji\prpg}}(t)^{\top}&B^{i\prpg}(t)\end{bmatrix}\!\!,\end{align*}
where the state-bias cross-covariance terms are $\vect{C}^{ii\prpg}=\mathrm{E}_{\text{f}}[(\vect{x}^i-\Hvect{x}^{i\prpg})(b^i-\hat{b}^{i\prpg})]$ and $\vect{C}^{ji\prpg}=\mathrm{E}_{\text{f}}[(\vect{x}^j-\Hvect{x}^{j\prpg})(b^i-\hat{b}^{i\prpg})]$. We note that $\vect{C}^{ji\prpg}\neq {\vect{C}^{ij\prpg}}^\top$, because $\vect{C}^{ij\prpg}=\mathrm{E}_{\text{f}}[(\vect{x}^i-\Hvect{x}^{i\prpg})(b^j-\hat{b}^{j\prpg})]$, where $b^j$ is the bias in the measurements taken by agent $j$. Finally, note that $B^{i\prpg}=\mathrm{E}[(b^i-\hat{b}^{i\prpg})(b^i-\hat{b}^{i\prpg})]=(\bar{\phi}^i)^2+\Phi^i$. Here, $\mathrm{E}_f[.]$ indicates that the expectation is taken over first-order approximate relative measurement or system models. 

We let every agent $i$ to maintain and propagate the set  $\{\vect{C}^{ij\prpg}\}_{i=1}^N$ of state-bias cross-covariances  between its local state and the bias in the measurements of all the agents $\{1,\cdots,N\}$, which is given by
\begin{align}
\label{eq::update_sb_correlation_dmv_prop}
\vect{C}^{il\prpg}(t+1)&=\vect{F}^{i}(t)\vect{C}^{il\updt}(t), \quad l\in\{1,\cdots,N\}, 
\end{align}
where $\vect{F}^{i}(t)=\!\partial \mathsf{f}^i(\Hvect{x}^{i\updt}(t-1),\vect{u}^{i}(t))/\partial{\vect{x}}^i$. Initially $\vect{C}^{il\updt}(0)=\vect{0}$, the state-bias cross-covariance terms become non-zero as agents update their states using inter-agent relative measurements. As seen in~\eqref{eq::update_sb_correlation_dmv_prop}, the propagated state-bias cross covariance terms of agent $i$ are computed locally. We show below also that these terms can be updated using local variables of agent $i$ and the state-bias cross-covaraince terms of agent $j$. Therefore, using the DMV type approach, we only need to account for lack of knowledge of $\vect{P}_{ij}^{\prpg}$, when we want to update states of agent $i$. 

We note that since~\eqref{eq::PJ-BPJ} holds, we can also write 
\begin{align*}
&
\vect{P}_{J}^{\prpg}(t)\leq
\begin{bmatrix}\frac{1}{\omega}\vect{P}^{i\prpg}(t)&\vect{0}&\vect{C}^{ii\prpg}(t)\\
\vect{0}&\frac{1}{1-\omega}\vect{P}^{j\prpg}(t)&\vect{C}^{ji\prpg}(t)\\{\vect{C}^{ii\prpg}}^{\top}(t)&{\vect{C}^{ji\prpg}}^{\top}(t)&B^{i\prpg}(t)\end{bmatrix}\,,~\omega\in[0,1].
\end{align*}
Then, by taking into account that in the SKF framework, agent $i$ updates its extended prior states according to~\eqref{eq::our-x_c} and 
\begin{align}\label{eq::our-x_c_biased}
    ~\hat{b}^{i\updt}(t)= \hat{b}^{i\prpg}(t),\quad B^{i\updt}(t)=B^{i\prpg}(t).
\end{align}
Note that 
\begin{align}\label{eq::dmv-SKF-bound}
&\mathrm{E}_{\text{f}}[(\vect{x}^i-\Hvect{x}^{i\updt})(\vect{x}^i-\Hvect{x}^{i\updt})^\top]\leq \Bvect{\mathsf{P}}^{i}(\omega,\vect{K}^{i})=\nonumber \\
&\begin{bmatrix}(\vect{I}\!-\!\vect{K}^i\vect{H}_i^i)&-\vect{K}^i\vect{H}_j^i&-\vect{K}^i
\end{bmatrix}\begin{bmatrix}\frac{1}{\omega}\vect{P}^{i\prpg}&\vect{0}&\vect{C}^{ii\prpg}\\
\vect{0}&\frac{1}{1-\omega}\vect{P}^{j\prpg}&\vect{C}^{ji\prpg}\\{\vect{C}^{ii\prpg}}^{\top}&{\vect{C}^{ji\prpg}}^{\top}&B^i\end{bmatrix}\nonumber\\
&\times\begin{bmatrix}(\vect{I}\!-\!\vect{K}^i\vect{H}_i^i)&-\vect{K}^i\vect{H}_j^i&-\vect{K}^i\end{bmatrix}^\top
\!\!+\vect{K}^i{R}^i{\vect{K}^i}^\top\!\!
\end{align}
for any $\omega\in[0,1]$, where we used the first-order expansion of ${h}(\vect{x}^{i},\vect{x}^{j})$ about $\Hvect{x}_J^{\prpg}$ described by $
  {h}(\vect{x}^i,\vect{x}^j)\approx\,
  {h}(\Hvect{x}^{i\prpg},\Hvect{x}^{j\prpg})\!+\!\vect{H}^i_{i}\,(\vect{x}^i\!\!-\!\Hvect{x}^{i\prpg})\!+\!\vect{H}^i_{j}\,(\vect{x}^j\!\!-\!\Hvect{x}^{j\prpg})+
(b^i-\hat{b}^i)$. The gain is found by minimizing the mean square error of the upper bound~\eqref{eq::dmv-SKF-bound} $\Bvect{\mathsf{K}}^i(\omega)=\underset{\vect{K}^i}{\argmin} \text{Tr}(\Bvect{\mathsf{P}}^{i}(\omega,\vect{K}^{i}))$, which gives us 
\begin{align}\label{eq::dmv-skf-gain_omega}
    \Bvect{\mathsf{K}}^i(\omega)=& (\frac{1}{\omega}\vect{P}^{i\prpg}{\vect{H}_i^i}^{\top}+\vect{C}^{ii\prpg}){\vect{S}_k^i}^{-1}.
\end{align}
where
\begin{align}\label{eq::sij}
    \vect{S}_k^l=&\,\vect{H}_{i}^i\frac{\vect{P}^{i\prpg}}{\omega}{\vect{H}_{i}^i}^\top\!\!\!+\!\vect{H}_{j}^i\frac{\vect{P}^{j\prpg}}{1\!-\!\omega}{\vect{H}_{j}^i}\!^\top+\vect{H}_i^i\vect{C}^{ii\prpg}+\vect{H}_j^i\vect{C}^{ji\prpg}\nonumber\\&+{\vect{C}^{ii\prpg}}^{\top}{\vect{H}_i^i}^{\top}+{\vect{C}^{ji\prpg}}^{\top}{\vect{H}_j^i}^{\top}+B^i+{R}^i
\end{align}

Using this gain, $\Bvect{\mathsf{P}}^{i}(\omega,\Bvect{\mathsf{K}}^{i}(\omega))$ in~\eqref{eq::dmv-SKF-bound} reads as 
\begin{align}\label{eq::P_l_omega}
    \Bvect{\mathsf{P}}^{i}(\omega)=\,&\Bvect{\mathsf{P}}^{i}(\omega,\Bvect{\mathsf{K}}^{i}(\omega))=\frac{\vect{P}^{i\prpg}}{\omega}-(\frac{\vect{P}^{i\prpg}}{\omega}{\vect{H}_i^i}^{\top}+\vect{C}^{ii\prpg})\times \nonumber\\
    &{\vect{S}_j^i}^{-1}(\frac{\vect{P}^{i\prpg}}{\omega}{\vect{H}_i^i}^{\top}+\vect{C}^{ii\prpg})^{\top}.
\end{align}
We obtain $\omega^i_\star$, the optimal $\omega\in[0,1]$,  from~\eqref{eq::omega_star} with  $\Bvect{\mathsf{P}}^{i}(\omega)$ given in~\eqref{eq::P_l_omega}.
Subsequently, the SKF based \textit{nloscorrectBelief} updated belief $\text{bel}^{i\updt}_{2}(t)=(\Hvect{x}^{i\updt}_{2}(t),\vect{P}^{i\updt}_{2}(t))$ for agent $i$ is 
\begin{align*} \Hvect{x}_{2}^{i\updt}&=\Hvect{x}^{i\prpg}+\vect{K}^i_{2}\,({z}^i_{j}-\hat{z}^i_{j}),\\
\vect{P}_{2}^{i\updt}&=\Bvect{\mathsf{P}}^{i}(\omega_\star^i),
\end{align*}
while the bias is updated according to~\eqref{eq::our-x_c_biased}. The $nloscorrentBelief$ corresponds to the model 2 based update ins Fig~\ref{fig::flow_imm_origin}. Moreover, the state-bias cross-covariances are updated according~to
\begin{align*}
\vect{C}^{il\updt}&=\mathrm{E}_{\text{f}}[(\vect{x}^i-\Hvect{x}^{i\updt})({b}^l-\hat{b}^{l\updt})]\\&=\mathrm{E}[((\vect{I}-\vect{K}^i_{2}\vect{H}_i^i)\Tvect{x}^{i\prpg}-\vect{K}^i_{2}\vect{H}_j^i\Tvect{x}^{j\prpg} - \vect{K}^i_{2}\tilde{b}^{i\prpg})\tilde{b}^{l\prpg}],
\end{align*}
where $\Tvect{x}^{k\prpg}= (\vect{x}^k-\Hvect{x}^{k\updt})$ and $\tilde{b}^{k\prpg}=({b}^k-\hat{b}^{k\prpg})$, $k\in\{1,\cdots,N\}$. Then, we can write
\begin{align*}
\vect{C}^{il\updt}&=
    \begin{cases}
    (\vect{I}-\vect{K}^i_{2}\vect{H}_i^i)\vect{C}^{ii\prpg}-\vect{K}^i_{2}\vect{H}_j^i\vect{C}^{ji\prpg} - \vect{K}^i_{2}{B}^{i\updt},& l=i\\
    (\vect{I}-\vect{K}^i_{2}\vect{H}_i^i)\vect{C}^{il\prpg}-\vect{K}^i_{2}\vect{H}_j^i\vect{C}^{jl\prpg},& l\neq i
    \end{cases}
\end{align*}
for any $l\in\{1,\cdots,N\}$. Here, $\vect{K}^i_{2}$ is given by~\eqref{eq::dmv-skf-gain_omega} evaluated at $\omega^i_\star$.

\textbf{\textit{PredictBias}  function:} This function given by~\eqref{eq::update_sb_correlation_dmv_prop} propagates the set  $\{\vect{C}^{il\prpg}\}_{l=1}^N$ of state-bias cross-covariances  between the local state of agent $i$ and the bias in the measurements of all the agents $\{1,\cdots,N\}$ locally.

\textbf{\textit{losProbability} and \textit{nlosProbability} functions:} These functions calculate the model $n\in\{1,2\}$ probability evolution in Fig.~\ref{fig::flow_imm_origin}, and their function is given by~\eqref{eq::imm_evolve}.

\textbf{\textit{combination} function:} This function realizes the last step in the IMM-based estimator of Fig.~\ref{fig::flow_imm_origin} given by~\eqref{eq::imm_combine}. Considering a Gaussian process, the combined belief $\text{bel}^{i\updt}(t)$ according to~\eqref{eq::imm_combine} is given by 
\begin{subequations}
\begin{align}\label{eq:update_combination_belief}
    \Hvect{x}^{i\updt}(t)=&\sum^{2}_{n=1}P(M_n(t)|\vect{Z}_{1:t})\Hvect{x}^{i\updt}_{n}(t),\\
    \vect{P}^{i\updt}(t)=&\sum^{2}_{n=1}P(M_n(t)|\vect{Z}_{1:t})(\vect{P}^{i\updt}_{n}(t)+\Bvect{P}^i_{n}(t)),
\end{align}
\end{subequations}
where $\Bvect{P}_{n}(t)=(\Hvect{x}^{i\updt}_{n}(t)-\Hvect{x}^{i\updt}(t))(\Hvect{x}^{i\updt}_{n}(t)-\Hvect{x}^{i\updt}(t))^{\top}$. The state-bias cross-covariance is affected by the combination of state and becomes
$\vect{C}^{il\updt}(t)=P(M_2|\vect{Z}_{1:t})\vect{C}^{il\updt}(t)$, $l\in\{1,\cdots,N\}$.

\textbf{Inter-agent communication}: To perform \textit{loscorrectBelief} function the local belief $\text{bel}^{j\prpg}$ of agent $j$ should be communicated to agent $i$. To preform \textit{nloscorrectBelief} function, besides the local belief $\text{bel}^{j\prpg}$, agent $j$ should transmit its  state-bias correlation set $col^{i\prpg}=\{\vect{C}^{jl\prpg}\}_{l=1}^N$ to agent $i$, as well.

\begin{rem}[Reducing the communication message size of the AUCL algorithm]{\rm
To reduce the communication message size/cost, we can allow agents to drop exact tracking of the inter-agent state-bias cross-covariance terms, and instead account for them implicitly. To do so, we write the joint extend state of agent $i$ and $j$ as $\vect{x}_J(t)\!=\!(\vect{x}^{i}(t)^\top\!,b^{i\prpg}(t), \vect{x}^{j}(t)^\top)^\top$, with the corresponding joint belief  $\text{bel}_J^{\prpg}(t)\!=\!(\Hvect{x}_J^{\prpg}(t),\vect{P}_J^{\prpg}(t))$, where $\vect{P}_{J}^{\prpg}(t)=\left[\begin{smallmatrix}
\vect{P}^{i\prpg}(t)&\vect{C}^{ii\prpg}(t)&\vect{P}_{ij}^{\prpg}(t)\\
\vect{C}^{ii}(t)^\top&B^{i\prpg}(t)&\vect{C}^{ji\prpg}(t)^\top\\
{\vect{P}_{ij}^{\prpg}(t)}^\top&\vect{C}^{ji\prpg}(t)&\vect{P}^{j\prpg}(t)
\end{smallmatrix}\right]$. Then to account for lack of knowledge about $\vect{C}^{ji\prpg}$, we use the upper bound on $\vect{P}_{J}^{\prpg}(t)$ in 
\begin{align}\label{eq::P_J-Cbound}
&
\vect{P}_{J}^{\prpg}(t)\leq
\begin{bmatrix}
\frac{1}{\omega}\begin{bmatrix}
\vect{P}^{i\prpg}(t)&\vect{C}^{ii\prpg}\\\vect{C}^{ii}(t)^\top&B^{i\prpg}(t)
\end{bmatrix}&\vect{0}\\
\vect{0}&\frac{1}{1-\omega}\vect{P}^{j\prpg}(t)
\end{bmatrix}\,,~\omega\in[0,1].
\end{align}
to obtain a $\Bvect{\mathsf{P}}^{i}(\omega,\vect{K}^{i})$ that satisfies $\mathrm{E}_{\text{f}}[(\vect{x}^i-\Hvect{x}^{i\updt})(\vect{x}^i-\Hvect{x}^{i\updt})^\top]\leq \Bvect{\mathsf{P}}^{i}(\omega,\vect{K}^{i})$ and does not depend on $\vect{P}_{ij}^{\prpg}$ and $\vect{C}^{ji\prpg}$. Then, we can obtain the update gain and the subsequent updates estimate and the covariance from a process similar to the one that follows~\eqref{eq::dmv-SKF-bound}.
Next, we note that 
\begin{align}\label{eq::dmv-SKF-bound_C}
&\mathrm{E}_{\text{f}}[(\vect{x}_J-\Hvect{x}_J^{\updt})(\vect{x}_J-\Hvect{x}_J^{\updt})^\top]\leq \nonumber \\
&(\vect{I}\!-\!\vect{K}_J\vect{H}_J)\begin{bmatrix}
\frac{1}{\omega^\star}\begin{bmatrix}
\vect{P}^{i\prpg}(t)&\vect{C}^{ii\prpg}\\\vect{C}^{ii}(t)^\top&B^{i\prpg}(t)
\end{bmatrix}&\vect{0}\\
\vect{0}&\frac{1}{1-\omega^\star}\vect{P}^{j\prpg}(t)
\end{bmatrix}\nonumber\\
&\quad \times(\vect{I}\!-\!\vect{K}_J\vect{H}_J)^\top+\vect{K}_J{R}^i{\vect{K}_J}^\top,
\end{align}
where $\vect{K}_J=
\begin{bmatrix}
{\vect{K}_2^i}^\top
&0
&\vect{0}\end{bmatrix}^\top$ and 
$\vect{H}_J=\begin{bmatrix}\vect{H}_i^i&1&\vect{H}_j^i\end{bmatrix}$.
Here, recall that $\vect{x}^{j\updt}(t)=\vect{x}^{j\prpg}(t)$, and $\hat{b}^{i\updt}(t)=\hat{b}^{i\prpg}(t)$.
Then to update the state-bias covariance for agent $i$, we use the corresponding component of the conservative upper bound in~\eqref{eq::dmv-SKF-bound_C}, which reads as  $\vect{C}^{ii\updt}(t)$ $\vect{C}^{ii\updt}(t)=\frac{1}{\omega^\star}(\vect{I}-\vect{K}^i_{2}\vect{H}_i^i)\vect{C}^{ii\prpg} - \frac{1}{\omega^\star}\vect{K}^i_{2}{B}^{i\updt}$. \boxend
}
\end{rem}

\textbf{Update with respect to beacons}: Similar as the relative range measurement with respect to a mobile agent, we follow the IMM estimator in Fig.~\ref{fig::flow_imm_origin} to process the range measurements with respect to beacons. Since the position of beacons are exactly known without the involvement of uncertainty, we simply employ the update step of EKF for LoS correction and employ the NLoS correction from our previous work~\cite{JZ-SSK:19sensor}. The only correlation needs to update is $\vect{C}^{ii\prpg}=\mathrm{E}_{\text{f}}[(\vect{x}^i-\Hvect{x}^{i\prpg})(b^i-\hat{b}^{i\prpg})]$.

We close this section with the following lemma that shows that if the measurement model is known deterministically, our proposed IMM-based CL gives the same updated estimate that the processing based on the known mode gives.   Therefore, we can conclude that the IMM-based CL is the more general method to treat the UWB ranging correction feedback.

\begin{lem}
\label{lemma::02}
If the measurement model at any time $t$ can be identified with absolute certainty, 
i.e., $p_n(t)=0$ or $1$, $n\in\{1,2\}$, the IMM-based CL update is equivalent to simply switch between \textit{loscorrectBelief} and \textit{nloscorrectBelief}.
\end{lem}
\begin{proof}
Consider the local belief $\text{bel}^{l\prpg}(t)$  for $l\in\{i,j\}$. Let the inter-agent range measurement $\vect{z}_j^i(t)$ detected at time $t$ be identified with absolute certainty as NLoS, i.e., $p_1(t)=0,\quad p_2(t)=1.$ 
Substituting~\eqref{eq::mix_1} into~\eqref{eq::imm_evolve}, we have that
\begin{align*}
        P(M(t)|\vect{Z}_{1:t})=&\frac{P(z_j^i(t)|M_n(t),\vect{Z}^i_{1:t})p_n(t)}{\sum_{m=1}^2P(z_j^i(t)|M_m(t),\vect{Z}^i_{1:t})p_m(t)},
\end{align*}
or $P(M_1(t)|\vect{Z}^i_{1:t})=0,\quad P(M_2(t)|\vect{Z}^i_{1:t})=1.$ Therefore, from~\eqref{eq:update_combination_belief}, we obtain
$
\text{bel}^{i\updt}(t)=\text{bel}^{i\updt}_2(t).
$ The same argument applies if the measurement is identified with absolute certainty to be in LoS.
\end{proof}

\section{Experimental evaluations}\label{sec::exp_eval}
We demonstrate the performance of our proposed AUCL algorithm via two experiments for a group of pedestrians who use UWB relative range measurements among themselves to improve their shoe-mounted INS system geolocation. The portable localization unit, shown in Fig.~\ref{fig:hardware}, that is used in these experiments consists of a foot-mounted IMU (VectorNav VN-100) and an UWB transceiver (DecaWave DWM1000) connected to a computing unit with a portable battery. 
  \begin{figure}
     \centering
     \includegraphics[scale=0.2]{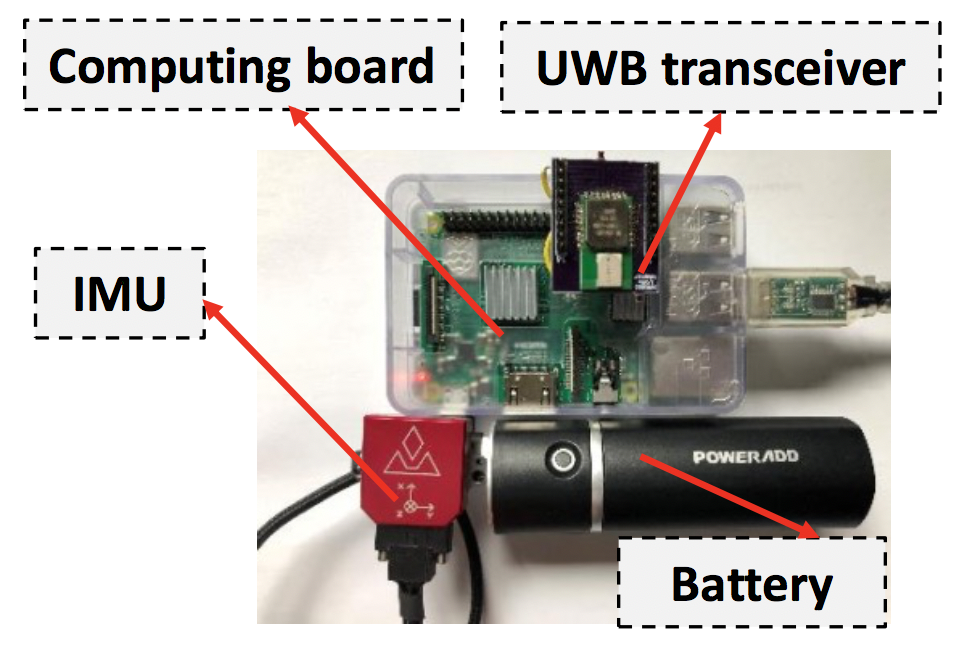}
     \caption{{\small The portable localization unit used in the experiment. The IMU mounts on the shoe.}}
     \label{fig:hardware}
 \end{figure}

Our first experiment was conducted on the second floor of the Engineering Gateway Building at the UCI campus (an indoor environment) with the floor plan that is shown in Fig.~\ref{fig:exp_res_second}. In this experiment, two pedestrians walked along a pre-defined reference trajectory shown by the black solid plot in Fig.~\ref{fig:exp_res_second}. They started from the black cross and went counter-clockwise. In this experiment, only agent $2$ has access to the beacon with a known position outside of the building. Beacon used to let agent $2$ have a better localization accuracy than agent $1$. Then, agent $1$ improves its own localization estimate by processing inter-agent range measurements with respect to agent $2$. The experiment demonstrates the benefit attained by cooperative localization, especially highlighting how access to absolute exteroceptive measurement by one agent can benefit others.
Since we know a priori that the beacon is outside of the building, and thus all the measurements between agent $2$ and the beacon should be in NLoS, agent $2$ processes measurements collected with respect to the beacon with $p_2(t)=1$. In this experiment, agent $2$ uses only the measurements from the beacon to improve its localization accuracy (red trajectory in Fig.~\ref{fig:exp_res_second}, with legend `Deterministic'). In case of inter-agent ranging between agent $1$ and $2$, we do not have any a priori knowledge about the UWB ranging mode at each time. Using a power-based UWB modal discriminator, the probability of the measurements between agents $1$ and $2$ being in NLoS is shown in the bottom left plot of Fig.~\ref{fig:exp_res_second}. As seen, by employing the AUCL algorithm, agent $1$ obtains a better localization (the blue trajectory with the legend `AUCL') in comparison to using a threshold to identify exactly the model of the inter-agent measurements and deterministic processing of the identified mode (the red plot with the legend `Deterministic'). The loop-closure errors in Fig.~\ref{fig:exp_res_second} are normalized by the length of the trajectory. The green trajectory with the legend `Naive UWB' shows the localization performance of a filter that ignores the bias in the NLoS measurements. As seen, in case of agent $2$ the performance the Naive UWB processing is even worse than the performance of INS only localization, because all the measurements between agent $2$ and the beacon are in NLoS and ignoring the bias in the measurements has a significant degrading effect. A video presentation of this experiment is available at~\cite{Demo_Youtube1}. 
 
  \begin{figure*}[!t]
     \centering
     \includegraphics[scale=0.4]{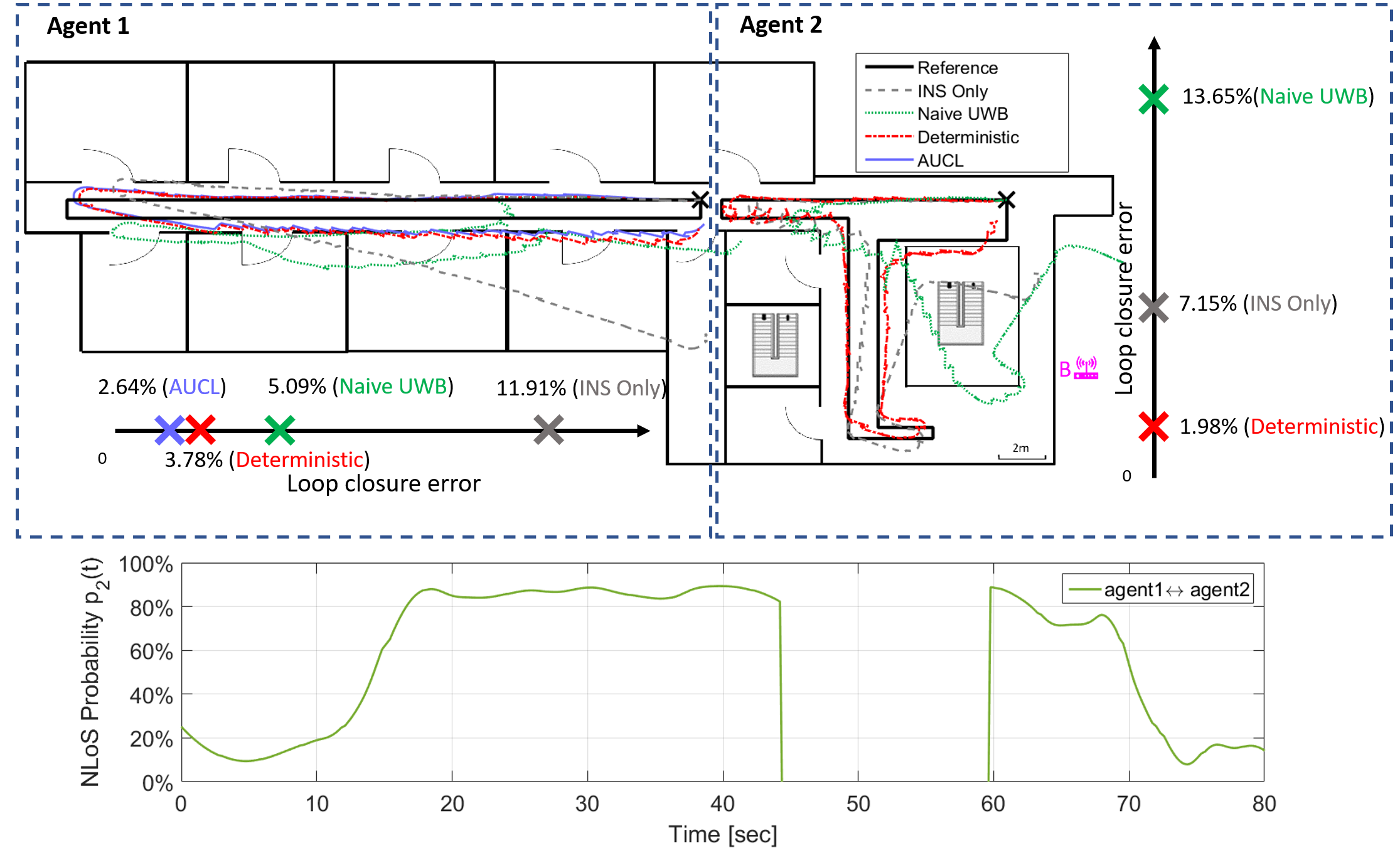}
     \caption{The localization result (trajectory and loop-closure error) of the first experiment in an indoor environment. The plot in the bottom shows the NLoS probability of inter-agent measurements and the gap in the plot is because the two agents were out of the sensing range of each other. }
     \label{fig:exp_res_second}
 \end{figure*}

\textbf{Second experiment}:
In our second experiment, a team of three pedestrian agents walked along a reference trajectory, which was in the outer path around the pool in Fig.~\ref{fig:exp_res} with a length of about $240$ meters.
 Each agent was equipped with the portable localization unit shown in Fig.~\ref{fig:hardware}. Two beacons (B1 and B2) were placed along the path at known locations as shown in Fig.~\ref{fig:hardware}. The inter-agent and the agents to the beacons range measurement mode was not known a priori and depending on where the agents were with respect to each other the measurement could be LoS or NLoS. The power-based modal discriminator was used to identify the probability of each measurement mode. The bottom three plots in Fig.~\ref{fig:exp_res} show the probability that the measurements are in NLoS during the test based on the power-based UWB modal discriminator. In this experiment, agent $2$ and agent $3$ started walking from the same point in the opposite direction. Agent $1$ waited along the path of agent $2$ and started later at the time when agent $2$ got closer. The experiment stopped when agent $2$ and agent $3$ returned to the starting point so we use the loop closure error of these two agents as our performance indicator.  We run four parallel localization filters on each agent. For all three agents, the INS only localization using the foot-mounted IMUs due to the error accumulation results in the trajectories that drift, as shown in the blue solid plot, with legend `INS only' in Fig.~\ref{fig:exp_res}. To bound the error, relative range measurements when agents were in the measurement range of each other were processed to update the local estimates obtained from INS. Due to the existence of obstruction in between agents such as bushes, trees, swimming pool equipment, and people, the measured relative range measurements were under a mix of LoS and NLoS conditions. Ignoring the bias in the measurements resulted in poor localization accuracy and even filter divergence as the black dotted plot with the legend `Naive UWB' in Fig.~\ref{fig:exp_res}. On the other hand, as seen in Fig.~\ref{fig:exp_res}, AUCL algorithm, the red plot with legend `AUCL', by employing bias compensation and also taking into account the probabilistic nature of the power-based UWB modal discriminator delivers the best localization and smallest loop closure error, which is expressed in terms of the percentage of the distance traveled. The trajectories in magenta with legend `Deterministic' shows the performance of the CL AUCL algorithm when we use deterministic identification using a threshold to identify the UWB measurement mode with absolute certainty ($p_2=0$ or $p_2=1$). As we can see ignoring the probabilistic nature of modal discriminator results in poorer localization performance. A video presentation of this experiment is available at~\cite{Demo_Youtube2}.

 \begin{figure*}[!t]
     \centering
     \includegraphics[scale=0.4]{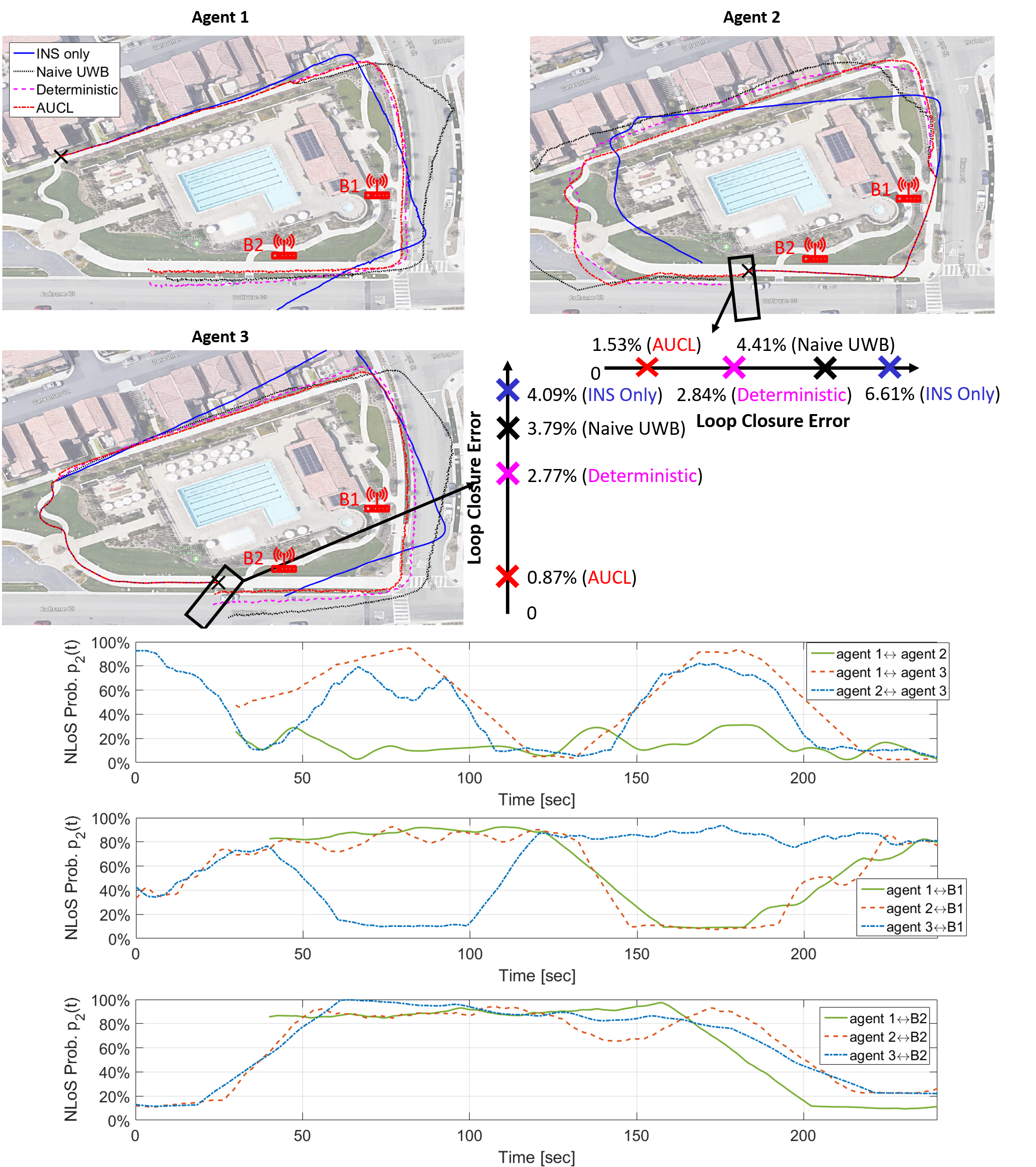}
     \caption{{\small The localization result (trajectories, and the loop-closure errors in terms of percentage of distance traveled) of a field testing with $3$ pedestrian agents and two beacons with known location. The three plots in the bottom show the NLoS probability of, respectively, inter-agent, agents to B1, and agents to B2 measurements.}}
     \label{fig:exp_res}
 \end{figure*}

\section{Conclusions}
We proposed an adaptive UWB based cooperative localization solution for applications where maintaining network-wide connectivity is challenging. 
Our design included a proper bias compensation for NLoS inter-agent UWB range processing, and also took into account the probabilistic multi-modal nature of UWB inter-agent range measurements. We used the IMM method to seamlessly handle the measurement model switching between LoS and NLoS in the UWB range measurements and used the Schmidt Kalman filtering for bias compensation.  We incorporated IMM filtering and bias compensation elements in the framework of a loosely coupled cooperative localization algorithm, that serves as an augmentation atop of a dead-reckoning system such as INS in a loose coupling manner. For each agent, this augmentation becomes active only when the agent takes a relative UWB range measurement with respect to another mobile agent or a beacon. To process the measurement, the agent needs only to communicate with the agent it has taken the measurements from. Our cooperative localization solution also is a practical sub-optimal solution with a low computational complexity, which can be implemented in real-time on a single computing board. We demonstrated the effectiveness of our method via a real-time localization of a pedestrian using an experimental setup.

\bibliographystyle{ieeetr}%
\bibliography{alias,references} 

\end{document}